\documentclass[aps,pra,showpacs,twocolumn]{revtex4}

\usepackage{amsmath}
\usepackage{amssymb}
\usepackage{mathrsfs}
\usepackage{amsthm}
\usepackage{bm}
\usepackage{url}
\usepackage[T1]{fontenc}
\usepackage{csquotes}
\MakeOuterQuote{"}
\newtheoremstyle{note}
  {\topsep}               % ABOVE SPACE
  {\topsep}               % BELOW SPACE
  {}                      % BODY FONT
  {\parindent}            % INDENT (empty value is the same as 0pt)
  {\itshape}              % HEAD FONT
  {.}                     % HEAD PUNCTUATION
  {5pt plus 1pt minus 1pt}% HEAD SPACE
  {}

\theoremstyle{note}
\newtheorem{theorem}{Theorem}
\newtheorem{lemma}{Lemma}
\newtheorem{conjecture}{Conjecture}

\theoremstyle{definition}

\theoremstyle{remark}

\newcommand{\tr}{\operatorname{tr}}
\newcommand{\Tr}{\operatorname{Tr}}

\newcommand{\be}{\begin{equation}}
\newcommand{\ee}{\end{equation}}
\newcommand{\ba}{\begin{align}}
\newcommand{\ea}{\end{align}}

\def\<{\langle}  %% overiding the original command \<
\def\>{\rangle}  %% overiding the original command \>

\newcommand{\dket}[1]{| #1\>\!\>}

\newcommand{\dbra}[1]{\<\!\< #1|}

\newcommand{\dinner}[2]{\<\!\< #1| #2\>\!\>}

\newcommand{\douter}[2]{| #1\>\!\>\<\!\< #2|}

\newcommand{\mse}{\mathcal{E}}

\newcommand{\barcal}[1]{\bar{\mathcal{#1}}}

\def\eqref#1{\textup{(\ref{#1})}}  %% overiding the original command \eqref
\newcommand{\eref}[1]{Eq.~\textup{(\ref{#1})}}

\newcommand{\esref}[1]{Eqs.~\textup{(\ref{#1})}}

\newcommand{\sref}[1]{Sec.~\ref{#1}}
\newcommand{\Sref}[1]{Section~\ref{#1}}

\newcommand{\thref}[1]{Theorem~\ref{#1}}

\newcommand{\Thsref}[1]{Theorems~\ref{#1}}

\newcommand{\lref}[1]{Lemma~\ref{#1}}

\newcommand{\cref}[1]{Conjecture~\ref{#1}}
\newcommand{\Cref}[1]{Conjecture~\ref{#1}}

\newcommand{\rcite}[1]{Ref.~\cite{#1}}
\newcommand{\rscite}[1]{Refs.~\cite{#1}}

\begin{document}
\title{Tomographic and Lie algebraic significance of  generalized symmetric informationally complete measurements}
\author{Huangjun Zhu}
\email{hzhu@pitp.ca}
\affiliation{Perimeter Institute for Theoretical Physics, Waterloo, Ontario, Canada N2L 2Y5}


%03.67.-a quantum information
%03.65.Wj quantum tomography, state reconstruction
%03.65.-w quantum mechanics
%06.20.Dk	Measurement and error theory
%02.10.De algebraic structure

%03.67.Mn Entanglement production, characterization and manipulation
%03.65.Ud Entanglement and quantum nonlocality
%(e.g. EPR paradox, Bell's inequalities, GHZ states, etc.)
%(for entanglement production in quantum information, see 03.67.Mn);

\begin{abstract}
Generalized symmetric informationally complete (SIC) measurements are SIC measurements that are not necessarily rank one. They are interesting originally because of their connection with rank-one SICs. Here we reveal several merits  of generalized SICs in connection with quantum state tomography and Lie algebra that are interesting in their own right.  These properties  uniquely characterize generalized SICs among minimal IC measurements although, on the face of it, they bear little  resemblance to the original definition.
In particular, we show that in  quantum state tomography generalized SICs are optimal  among minimal IC measurements with given average purity of  measurement outcomes. Besides its significance to the current study, this result may help  understand  tomographic efficiencies of minimal IC measurements under the influence of noise.
When minimal IC measurements are taken as  bases for the Lie algebra of the unitary group,  generalized SICs are uniquely characterized by the antisymmetry of the associated structure constants.
\end{abstract}

\date{\today}
\maketitle

\section{Introduction}

Quantum state tomography is a primitive of various quantum information processing tasks, such as quantum computation, communications and cryptography. To achieve high tomographic efficiency in  practice,  it is crucial to choose suitable measurements.
Of special interest are the type of
 measurements that are  \emph{informationally complete} (IC) with which every state can be determined completely by the measurement statistics.  An IC   measurement has at least $d^2$ outcomes for a $d$-level quantum system; those with  $d^2$ outcomes are called \emph{minimal}.

A \emph{symmetric informationally complete} (SIC) measurement \cite{Zaun11, ReneBSC04, Rene04the,Appl05,ScotG10, Zhu12the, Zhu10, ApplFZ14G}
is composed of $d^2$ subnormalized projectors onto pure states
$\Pi_j=|\psi_j\rangle\langle\psi_j|/d$ with equal pairwise inner product of $1/(d+1)$,
\begin{equation} \label{eq:SICinner}
|\langle\psi_j|\psi_k\rangle|^2=\frac{d\delta_{jk}+1}{d+1},\quad
j,k=0,1,\cdots,d^2-1.
\end{equation}
SICs  possess  many nice properties that make them an ideal choice of fiducial measurements.
For example,   they are
optimal for linear quantum state tomography
\cite{Scot06,ZhuE11,Zhu14IOC} and measurement-based quantum cloning~\cite{Scot06}. They  play a crucial role in  quantum Bayesianism
\cite{FuchS13,Fuch10,ApplEF11}.  They are also  interesting because of their  connections with mutually unbiased bases (MUB)
\cite{Ivan81,WootF89,DurtEBZ10,Woot06,ApplDF07}, 2-designs \cite{Zaun11, ReneBSC04,Scot06,ApplFZ14G}, equiangular
lines~\cite{Zaun11,ApplFZ14G},
Galois theory~\cite{ApplAZ13}, Lie algebra~\cite{ApplFF11, ApplFZ14G}, adjoint representation of the unitary group~\cite{ApplFZ14G}, and the graph isomorphism problem~\cite{Zhu12the}.
Up to now, analytical solutions of SICs and numerical solutions with high precision have been found up to dimension 67
\cite{Zaun11, ReneBSC04,Rene04the, Appl05,ScotG10,Zhu12the}.
Although SICs are expected to exist for  every finite dimension,  there is neither a universal construction recipe  nor a general existence proof~\cite{ApplFZ14G}.

 Generalized SICs are SICs whose outcomes are not necessarily rank one. They were first studied systematically by Appleby~\cite{Appl07},
and have raised some renewed interest recently~\cite{KaleG13}. Unlike  rank-one SICs, their existence is almost trivial, as the existence of regular tetrahedra, and several explicit construction methods are known~\cite{Appl07,KaleG13}.
  Nevertheless, the study of generalized SICs may promote our understanding about rank-one SICs and provide valuable  insight on the SIC existence problem.
They are also of interest from a practical point of view since measurements realized in experiments are usually not rank one due to various imperfections, such as noise and dark counts. It is thus highly desirable to determine whether generalized SICs retain some nice properties of rank-one SICs and whether they are optimal for some quantum information processing tasks under such scenarios.

In this paper we reveal several  nice properties of generalized SICs in connection with quantum state tomography and Lie algebra.
Remarkably, these properties  uniquely characterize generalized SICs among minimal IC measurements although they do not bear any resemblance to the  original definition.
In particular, we show that generalized SICs are optimal in  quantum state tomography with minimal IC measurements given the average purity of the measurement outcomes. Besides its significance to the current study, this result is pretty useful  in determining  the impact of noise on the tomographic efficiency. The outcomes of a minimal IC measurement can also  serve as a basis for the Lie algebra of the unitary group. We show that the structure constants associated with this basis are completely antisymmetric if and only if the measurement is a generalized SIC measurement. This observation generalizes the link between SICs and Lie algebra established in \rscite{ApplFF11,ApplFZ14G}. In the course of our study, we  derive several   useful results  about IC measurements, which may be of independent interest. Our study also leads to an intriguing connection between SICs and MUB popping up in   a tomography problem.

The rest of the paper is organized as follows. In \sref{sec:QST} we review the basic framework of  quantum state tomography. In \sref{sec:TightIC} we reexamine  tight IC measurements originally introduced by  Scott \cite{Scot06}. In \sref{sec:Balanced} we introduce the concept of balanced measurements and propose a conjecture about SICs and MUB. In \sref{sec:TomoSignificance} we reveal tomographic significance of generalized SICs as well as their connections with tight IC measurements and balanced measurements. In \sref{sec:LieSignificance} we present a cute characterization of generalized SICs in terms of  Lie algebra. \Sref{sec:Summary} summarizes this paper.

\section{\label{sec:QST}Quantum state tomography}

In this section, we review the basic framework of  quantum state tomography  tailored to the needs of the current work following \rscite{Scot06, ZhuE11, Zhu12the, Zhu14IOC}, in preparation for later discussions.

A  generalized measurement is composed of a set of outcomes
represented mathematically by positive operators $\Pi_j $ that sum up
to the identity  1. Given an unknown  state $\rho$, the
probability of obtaining the outcome $\Pi_j $ is given by the Born
rule: $p_j =\tr(\Pi_j \rho)$. Following the convention in \rscite{ZhuE11, Zhu12the, Zhu14IOC}, the probability can be expressed as an inner product $\dinner{\Pi_j }{\rho}$ between the operator kets $\dket{\Pi_j}$ and $\dket{\rho}$, where the double ket notation is used to distinguish them from ordinary kets. A measurement is IC if the outcomes $\Pi_j$ span the whole operator space.

For an IC measurement  $\{\Pi_j\}$,
there exists  a set of
reconstruction operators $\Theta_j$
such that
$\sum_j\douter{\Theta_j}{\Pi_j}=\mathbf{I}$,
where $\mathbf{I}$ is the identity superoperator. Any state can be recovered from the set of
probabilities $p_j$ using the formula $\rho=\sum_j p_j\Theta_j$. In practice, the probabilities $p_j$ need to be replaced by the frequencies $f_j$  since the number  of measurements is finite. The estimator based on these frequencies
$\hat{\rho}=\sum_j f_j\Theta_j$ is thus different from the true
state.  Nevertheless, the requirement $\sum_j\douter{\Theta_j}{\Pi_j}=\mathbf{I}$ on the reconstruction operators guarantees that the estimator is unbiased. The scaled MSE
matrix and MSE with respect to the Hilbert-Schmidt (HS) distance of the estimator $\hat{\rho}$ are given by~\rscite{ZhuE11, Zhu14IOC},
\begin{align}
\mathcal{C}(\rho)&=\sum_j\dket{\Theta_j}p_j\dbra{\Theta_j}-\douter{\rho}{\rho},       \label{eq:MSEmatrix}\\
\mse(\rho)&=\Tr\{\mathcal{C}(\rho)\}=\sum_j p_j\tr\bigl(\Theta_j^2\bigr)-\tr(\rho^2).  \label{eq:MSEg}
\end{align}
Here "$\Tr$"  denotes the trace of a superoperator,
and "$\tr$" of an ordinary operator.

The set of reconstruction operators is not unique except for a minimal IC
measurement. In linear state tomography, usually these  operators, once chosen, are independent of the measurement
statistics.  In that case, the average scaled MSE over unitarily equivalent states is given by
\begin{align}\label{eq:aMSEg}
\overline{\mse(\rho)}&=\frac{1}{d}\sum_j
\tr(\Pi_j)\tr\bigl(\Theta_j^2\bigr)-\tr(\rho^2).
\end{align}
 \emph{Canonical reconstruction
operators}
\begin{equation}\label{eq:CanonicalR}
\dket{\Theta_j}=\frac{d\mathcal{F}^{-1}\dket{\Pi_j}}{\tr(\Pi_j)}
\end{equation}
are the best  for minimizing the average scaled MSE \cite{Scot06,ZhuE11,Zhu12the, Zhu14IOC}, where
\begin{equation}\label{eq:FrameSO1}
\mathcal{F}=d\sum_j \frac{\douter{\Pi_j }{\Pi_j }}{\tr(\Pi_j)}
\end{equation}
is known as the frame superoperator. The minimum reads
\begin{align}\label{eq:aMSEcanonical}
\overline{\mse(\rho)}&:=\Tr(\mathcal{F}^{-1})-\tr(\rho^2).
\end{align}

If reconstruction operators are allowed to depend on the measurement statistics, the  \emph{optimal reconstruction operators}  are given by
\begin{equation}\label{eq:OptimalR}
\dket{\Theta_j}=p_j^{-1}\mathcal{F}(\rho)^{-1}\dket{\Pi_j},
\end{equation}
where
\begin{equation}\label{eq:FrameSOopt}
\mathcal{F}(\rho)=\sum_{j}\dket{\Pi_j}\frac{1}{p_j}\dbra{\Pi_j}
\end{equation}
is also called the frame superoperator, which generalizes the definition in \eref{eq:FrameSO1}~\cite{Zhu12the,Zhu14IOC}. Unlike usual linear state tomography, the optimal reconstruction operators depend on the unknown state and need to be chosen adaptively in practice.
The scaled MSE matrix and MSE turn out to be
\begin{align}
\mathcal{C}(\rho)&=\mathcal{F}(\rho)^{-1}-\douter{\rho}{\rho}=\barcal{F}(\rho)^{+},  \label{eq:BlueMSEmatrix}\\
\mse(\rho)&=\Tr\bigl\{\mathcal{F}(\rho)^{-1}\bigr\}-\tr(\rho^2)=\Tr\bigl\{\mathcal{F}(\rho)^{+}\bigr\},   \label{eq:BlueMSE}
\end{align}
where $\barcal{F}(\rho)$ is the projection of $\mathcal{F}(\rho)$ onto the space of traceless Hermitian operators and is the superoperator analog of the Fisher information matrix~\cite{Fish25}; $\barcal{F}(\rho)^{+}$ denotes the Moore-Penrose pseudoinverse of $\barcal{F}(\rho)^{+}$, that is, the inverse on its support. Therefore, the above equations actually give the Cram\'er-Rao bounds~\cite{Cram46, Rao45} for the scaled MSE matrix and MSE~\cite{Zhu12the,Zhu14IOC}.

For a minimal IC measurement, the  set of reconstruction operators is unique, so \esref{eq:BlueMSEmatrix} and \eqref{eq:BlueMSE} reduce to \esref{eq:MSEmatrix} and \eqref{eq:MSEg} with $\Theta_j$  canonical reconstruction operators.
When $\rho$ is the completely mixed state, that is $\rho=1/d$, \esref{eq:OptimalR} and  \eqref{eq:FrameSOopt}
reduce to \esref{eq:CanonicalR} and \eqref{eq:FrameSO1}, so the optimal reconstruction is also identical with the  canonical reconstruction. The scaled MSE matrix and MSE are respectively given by
\begin{align}
\mathcal{C}(\rho)&=\mathcal{F}^{-1}-\frac{1}{d^2}\douter{1}{1}=\barcal{F}^+,  \label{eq:BlueMSEmatrixCM}\\
\mse(\rho)&=\Tr(\mathcal{F}^{-1})-\frac{1}{d}=\Tr(\barcal{F}^+).   \label{eq:BlueMSECM}
\end{align}

\section{\label{sec:TightIC}Tight informationally complete measurements}
Tight IC measurements were first introduced by Scott~\cite{Scot06} as measurements  featuring particular simple state reconstruction. Rank-one tight IC measurements are also optimal under linear quantum state tomography and have thus attracted much attention recently \cite{RoyS07, ZhuE11, Zhu12the, Zhu14IOC,ApplFZ14G}. General tight IC measurements are still not well understood, although they are more relevant in real experiments. In this section, we
derive several useful results  about these  measurements, thereby deepening our understanding on this subject.
In particular, we  determine the minimal average MSE achievable  in linear quantum state tomography  given the average purity of  measurement outcomes and show that the minimum is attained only for tight IC measurements.
We also provide an alternative characterization of minimal tight IC measurements, which is quite useful for  understanding their structure and their connections with generalized SICs.

A measurement $\{\Pi_j\}$ is \emph{tight IC} \cite{Scot06, ZhuE11, Zhu12the,ApplFZ14G} if the frame superoperator has the following form
\begin{equation}\label{eq:TightIC}
\mathcal{F}=d\sum_j\frac{\douter{\Pi_j}{\Pi_j}}{\tr(\Pi_j)}=\alpha \mathbf{I}+\beta\douter{1}{1}
\end{equation}
for some positive constants $\alpha$ and $\beta$. Multiplying the equation by $\dket{1}$ on the right gives $\alpha+d\beta=d$, that is, $\beta=1-\alpha/d$. Taking trace of the equation yields
\begin{equation}
d^2\alpha+d\beta=d\sum_j\frac{\tr(\Pi_j^2)}{\tr(\Pi_j)}=d^2\sum_j \frac{\tr(\Pi_j)}{d} \wp_j=d^2 \wp,
\end{equation}
where $\wp_j=\tr(\Pi_j^2)/[\tr(\Pi_j)]^2$ is the purity of the outcome~$\Pi_j$ and $\wp=\sum_j\tr(\Pi_j) \wp_j/d$ can be understood as the average purity of the measurement $\{\Pi_j\}$. Both $\alpha$ and~$\beta$ are functions of the average purity,
\begin{equation}\label{eq:alphabeta}
\alpha=\frac{d^2\wp-d}{d^2-1},\quad \beta=\frac{d^2-d\wp}{d^2-1}.
\end{equation}
This equation implies that $\alpha\leq d/(d+1)$ and the inequality is saturated if and only if all $\Pi_j$ have rank one.

For a tight IC measurement satisfying \eref{eq:TightIC}, the  canonical reconstruction operators have a simple form,
\begin{equation}\label{eq:ReconTightIC}
\Theta_j =\frac{d^2\Pi_j-(d-\alpha)\tr(\Pi_j)}{d\alpha\tr(\Pi_j)},
\end{equation}
with
\begin{equation}\label{eq:ReconTightICnorm}
\tr(\Theta_j^2)=\frac{d^2\wp_j-d}{\alpha^2}+\frac{1}{d}.
\end{equation}
According to \eref{eq:MSEg}, the scaled  MSE associated with the canonical linear reconstruction is
\begin{equation}\label{eq:MSEtightIC}
\mse(\rho)=\frac{d^2}{\alpha^2}\Bigl[\wp(\rho)-\frac{1}{d}\Bigr]-\Bigl[\tr(\rho^2)-\frac{1}{d}\Bigr],
\end{equation}
where $\wp(\rho)=\sum_j p_j \wp_j$ is the average purity of the outcomes $\Pi_j$ weighted by the probabilities $p_j$. Note that $\wp(\rho)$ reduces to the average purity $\wp$ of $\{\Pi_j\}$ when $\rho$ is the completely mixed state; the average of $\wp(\rho)$ over unitarily equivalent states is also equal to $\wp$. Taking average in the above equation yields
\begin{equation}\label{eq:aMSEtightIC}
\overline{\mse(\rho)}=\frac{(d^2-1)^2}{d^2\wp-d}-\Bigl[\tr(\rho^2)-\frac{1}{d}\Bigr],
\end{equation}
where we have applied \eref{eq:alphabeta}.

The tomographic significance of tight IC measurements is revealed by the following theorem, which sets the tomographic efficiency limit of linear quantum state tomography in terms of the average purity of measurement outcomes.
\begin{theorem}\label{thm:TightICeff}
In linear quantum state tomography with any  IC measurement  with average purity $\wp$, the average scaled MSE over unitarily equivalent states is lower bounded by \eref{eq:aMSEtightIC}. The  bound is saturated if and only if the measurement is tight IC.
\end{theorem}

\begin{proof}
Let $\{\Pi_j\}$ be an IC measurement  with average purity $\wp$. The minimum average scaled MSE under linear tomography is given by $\Tr(\mathcal{F}^{-1})-\tr(\rho^2)$ according to \eref{eq:aMSEcanonical}. Note that $\Tr(\mathcal{F})=d^2\wp$ and that  $\dket{1}$ is an eigenvector of $\mathcal{F}$ with eigenvalue $d$. The minimum of $\Tr(\mathcal{F}^{-1})$ under these constraints is attained if and only if $\mathcal{F}$ has the form in \eref{eq:TightIC}, that is, if $\{\Pi_j\}$ is tight IC. In that case,    the lower bound for the average scaled MSE is indeed saturated.
\end{proof}

The minimum of $\overline{\mse(\rho)}$  is attained when $\wp=1$, that is, when the tight IC measurement is rank one.
So rank-one tight IC measurements are optimal in linear quantum state tomography  \cite{Scot06, ZhuE11, Zhu14IOC}.
The frame superoperator and the reconstruction operators now simplify to
\begin{equation}
\mathcal{F}=\frac{d}{d+1} (\mathbf{I}+\douter{1}{1}),\quad
\Theta_j =(d+1)\frac{\Pi_j}{\tr(\Pi_j)} -1.
\end{equation}
The scaled MSE is given by
\begin{equation}\label{eq:MSEtightICpure}
\mse(\rho)=d^2+d-1-\tr(\rho^2),
\end{equation}
which is  unitarily invariant. According to Scott~\cite{Scot06} (see also~\rcite{ApplFZ14G}), a rank-one measurement is tight IC if and only if the outcomes form a weighted 2-design.  Prominent examples of tight IC measurements are SIC measurements and complete mutually unbiased measurements, that is, measurements composed of complete sets of MUB. In the second  example, the scaled MSE can be reduced by the optimal reconstruction \cite{RehaH02,EmbaN04, Zhu12the,Zhu14IOC}, with the result
\begin{equation}
\mse(\rho)=d^2+d-(d+1) \tr(\rho^2).
\end{equation}
It is still unitarily invariant, which is quite rare among informationally overcomplete  measurements.

It is not easy to understand the  structure of tight IC measurements from the definition. The following lemma gives an alternative characterization of minimal tight IC measurements, which  is useful for understanding their structure and their connections with generalized SICs.
\begin{lemma}\label{lem:TightICinner}
A minimal IC measurement $\{\Pi_j\}$ is tight IC if and only if it satisfies the equation
\begin{equation}\label{eq:TightIC2}
\tr(\Pi_j\Pi_k)=\tilde{\alpha} \sqrt{\tr(\Pi_j)\tr(\Pi_k)}\delta_{jk}+\tilde{\beta} \tr(\Pi_j)\tr(\Pi_k)
\end{equation}
for some positive constants $\tilde{\alpha}$ and $\tilde{\beta}$.
\end{lemma}
\begin{proof}
Define
$L_j:=\Pi_j/\sqrt{\tr(\Pi_j)}$, then $\{\Pi_j\}$ satisfies  \eref{eq:TightIC} if and only if $\{L_j\}$ satisfies
\begin{equation}
d\sum_j\douter{L_j}{L_j}=\alpha \mathbf{I}+\beta\douter{1}{1}.
\end{equation}
According to Theorem 1 in \rcite{ApplFZ14G},  this equation is equivalent to
\begin{equation}
\tr(L_j L_k) =\tilde{\alpha} \delta_{jk} +\tilde{\beta} \tr(L_j)\tr(L_k),
\end{equation}
where $\tilde{\alpha}=\alpha/d$ and $\tilde{\beta}=\beta/(\alpha+d\beta)$. Replace $L_j$ with $\Pi_j/\sqrt{\tr(\Pi_j)}$ in the equation yields
\begin{equation}
\frac{\tr(\Pi_j\Pi_k)}{\sqrt{\tr(\Pi_j)\tr(\Pi_k)}}=\tilde{\alpha} \delta_{jk} +\tilde{\beta} \sqrt{\tr(\Pi_j)\tr(\Pi_k)},
\end{equation}
which is equivalent to \eref{eq:TightIC2}. So \eref{eq:TightIC} is equivalent to \eref{eq:TightIC2} with
$\tilde{\alpha}=\alpha/d$ and $\tilde{\beta}=\beta/(\alpha+d\beta)$. Since $\alpha+d\beta=d$ for a tight IC measurement, it follows that $\tilde{\beta}=\beta/d$.
\end{proof}

\section{\label{sec:Balanced}Balanced measurements}
An IC measurement is \emph{quasi-balanced} if  the scaled MSE $\mse(\rho)$ [which is equal to the Cram\'er-Rao bound;  see \eref{eq:BlueMSE}] of the optimal reconstruction  is unitarily invariant. It is \emph{balanced} if in addition the scaled MSE matrix [see~\eref{eq:BlueMSEmatrixCM}] at the completely mixed state is invariant under unitary transformations of the measurement outcomes \footnote{The balanced measurements defined in the author's thesis~\cite{Zhu12the} correspond to quasi-balanced measurements here.}.  Intuitively, balanced measurements are those measurements whose tomographic efficiencies are most indifferent to the identity of the true state. This concept also has an intimate connection with operationally invariant information proposed in~\rcite{BrukZ99} (cf.~\rcite{RehaH02}).

According to \eref{eq:BlueMSE}, an IC measurement is quasi-balanced if and only if $\Tr\{\mathcal{F(\rho)}^{-1}\}$ is unitarily invariant.
According to \eref{eq:BlueMSEmatrixCM}, the additional requirement for a balanced measurement is satisfied if and only if the frame superoperator $\mathcal{F}$ is unitarily invariant. Since the adjoint representation of the unitary group has only two irreducible components, one spanned by the identity operator $1$ and the other  by traceless Hermitian operators, this requirement is satisfied if and only if $\mathcal{F}$ has the form as in \eref{eq:TightIC}, that is, if the measurement is tight IC. Therefore, a balanced  measurement is one that  is both tight IC and quasi-balanced.

According to \sref{sec:TightIC} (see also \rscite{Zhu12the, Zhu14IOC}) and the previous discussion, SIC measurements and complete mutually unbiased measurements are balanced. Actually,
 they are the  only known rank-one balanced measurements with finite number of outcomes (assuming different outcomes of a measurement are not proportional to each other; the covariant measurement is balanced but with infinite number of outcomes).
 It is plausible that there is a  nontrivial connection between SICs and MUB underlying this coincidence.
\begin{conjecture}
SIC measurements and complete mutually unbiased measurements are the only rank-one balanced measurements with finite number of outcomes.
\end{conjecture}
To appreciate the difficulty in constructing balanced measurements, note that combinations of balanced measurements are generally not balanced. For example, the cube measurement in the case of a qubit is not balanced although it is composed of two SIC measurements~\cite{Zhu14IOC}.

 In general, it is not easy to characterize all quasi-balanced measurements. For a minimal IC measurement, since the set of  reconstruction operators is unique, the minimal scaled  MSE  $\mse(\rho)$ is  determined by \eref{eq:MSEg}, where $\Theta_j$ are canonical reconstruction operators. It  is unitarily invariant if and only if  $\sum_j p_j \tr(\Theta_j^2)$ is unitarily invariant and, consequently, independent of $\rho$. This is possible if  $\tr(\Theta_j^2)$ is independent of $j$ and only then.

\begin{lemma}
A minimal IC measurement is quasi-balanced if and only if all reconstruction operators have the same HS norm.
\end{lemma}
According to this lemma, any group covariant minimal IC measurement is quasi-balanced since all reconstruction operators have the same spectrum due to symmetry.

For a minimal tight IC measurement, according to \eref{eq:ReconTightIC} or \eqref{eq:ReconTightICnorm},  reconstruction operators have the same HS norm if and only if  outcomes have the same purity.
\begin{lemma}\label{lem:TightBalance}
A minimal tight IC measurement is balanced if and only if all outcomes have the same purity.
\end{lemma}
Alternatively, this lemma follows from \eref{eq:MSEtightIC}.

\section{\label{sec:TomoSignificance}Tomographic significance of generalized SICs}

In this section we reveal several  tomographic merits of  generalized SICs after a short introduction. In particular, we show that among minimal IC measurements, they are identical with balanced measurements and  are optimal in quantum state tomography given the average purity of measurement outcomes. Our study generalizes the result of Scott~\cite{Scot06} that SICs are optimal minimal IC measurements for linear quantum state tomography.

\subsection{Generalized SICs}

A measurement $\{\Pi_j=P_j/d\}$ with $n$ elements is called a generalized SIC~\cite{Appl07,KaleG13} if it is IC and satisfies
\begin{equation}\label{eq:GenSIC}
\tr(P_jP_k)=\alpha\delta_{jk}+\zeta
\end{equation}
for some real constants $\alpha$ and $\zeta$. The IC requirement implies that the Gram matrix of $\{P_j\}$ has  rank $d^2$, so that $\alpha>0$ and $n=d^2$.  Summing over $k$ in \eref{eq:GenSIC} yields $\alpha+d^2\zeta=d\tr(P_j)=d$. So all outcomes $\Pi_j$  have the same trace of $1/d$ and the same purity
\begin{equation}
\wp=\frac{\tr(\Pi_j^2)}{[\tr(\Pi_j)]^2}=\tr(P_j^2)= \frac{(d^2-1)\alpha +d}{d^2}.
\end{equation}
Consequently,
\begin{equation}\label{eq:alphabetaGenSIC}
\alpha=\frac{d^2\wp-d}{d^2-1},\quad \zeta=\frac{d-\wp}{d^2-1}.
\end{equation}
It follows that  $\alpha\leq d/(d+1)$ and the upper bound is saturated if and only if the generalized SIC is rank one.
Note that the expression for $\alpha$ and its range are the same as that for a tight IC measurement; cf.~\eref{eq:alphabeta}.

According to the above discussion, the outcomes of any generalized SIC can be written as \begin{equation}
\Pi_j=\frac{1}{d^2}(1+B_j),
\end{equation}
where the $B_j$ form a regular simplex in the space of traceless Hermitian operators. Conversely, any such regular simplex defines a generalized SIC as long as the minimum eigenvalues of $B_j$ are bounded from below by $-1$, as noticed by Appleby~\cite{Appl07}. In principle,  this observation allows constructing all generalized SICs.  For example, any generalized SIC in dimension 2 has the form
\begin{equation}
\Pi_j=x \tilde{\Pi}_j +\frac{1-x}{d^2},
\end{equation}
 where $\{\tilde{\Pi}_j \}$ is a rank-one SIC and $0< x\leq 1$.
 Unfortunately, in general, it is not clear at all whether the set of generalized SICs so constructed contains a rank-one SIC. An alternative construction (with the same limitation) was recently proposed in  \rcite{KaleG13}.

\subsection{Tomographic significance}
According to \lref{lem:TightICinner}, any generalized SIC is a tight IC measurement. If the measurement $\{\Pi_j\}$ satisfies \eref{eq:GenSIC}, then it also satisfies  \eref{eq:TightIC} with the same $\alpha$ and $\beta=d\zeta$. Now \lref{lem:TightBalance} implies that it is also a balanced measurement.
What is remarkable is that  the converse holds for any minimal IC measurement.
\begin{theorem}\label{thm:BalanceGenSIC}
A minimal IC measurement is balanced if and only if it is a generalized SIC measurement.
\end{theorem}
Before proving this theorem, we first point out its main implications.
As an immediate consequence, a rank-one minimal IC measurement is balanced if and only if it is SIC.
The scaled MSE achievable with a generalized SIC is given by \eref{eq:aMSEtightIC} without the line over $\mse(\rho)$, that is,
\begin{equation}\label{eq:MSEgenSIC}
\mse(\rho)=\frac{(d^2-1)^2}{d^2\wp-d}-\Bigl[\tr(\rho^2)-\frac{1}{d}\Bigr].
\end{equation}
 \Thsref{thm:TightICeff} and \ref{thm:BalanceGenSIC} together imply that
\begin{theorem}\label{thm:GenSICeff}
In  quantum state tomography with a minimal  IC measurement  with average purity $\wp$, the maximal scaled  MSE of any unbiased estimator over unitarily equivalent states is bounded from below by \eref{eq:MSEgenSIC}. The  bound can be saturated if and only if the measurement is a generalized SIC.
\end{theorem}
 Like \thref{thm:TightICeff}, this theorem sets the tomographic efficiency limit  in terms of the average purity of measurement outcomes, except that  the figure of merit is the maximal scaled MSE instead of the average scaled MSE. Besides, we are not restricted to linear estimators, because  the Cram\'er-Rao bound for the scaled MSE is saturated in canonical linear state tomography with a minimal IC measurement.
 In addition to providing neat characterizations of tight IC measurements and generalized SIC measurements, the two theorems are quite useful in studying tomographic efficiencies of minimal IC measurements. As an implication of \thref{thm:GenSICeff},  the maximal scaled  MSE with  minimal  IC measurements is bounded from below by \eref{eq:MSEtightICpure}, and the bound can be saturated only for rank-one SIC measurements~\cite{Scot06}.

According to \lref{lem:TightBalance}, to prove \thref{thm:BalanceGenSIC}, it suffices to show that a minimal tight IC measurement is a generalized SIC measurement if and only if all outcomes have the same purity, which follows from the following lemma.
\begin{lemma}\label{lem:TightICgenSIC}
A minimal tight IC measurement $\{\Pi_j\}$ is a generalized SIC measurement if and only if any of the following conditions is satisfied:
\begin{enumerate}
\item $\tr(\Pi_j)$ is independent of $j$.

\item $\tr(\Pi_j^2)$ is independent of $j$.

\item The purity of $\Pi_j$ is independent of $j$.

\item $\{\Pi_j\}$ is equiangular.
\end{enumerate}
\end{lemma}
For a rank-one minimal IC measurement, the purities of all outcomes are automatically identical. Therefore, it is tight IC if and only if it is SIC~\cite{Scot06, ApplFZ14G}.

\begin{proof}
Obviously,  statements 1, 2, 3, and 4 hold for a generalized SIC. By \lref{lem:TightICinner}, the four statements are equivalent for a tight IC measurement, and any of them guarantees that   $\{\Pi_j\}$ is a generalized SIC.
\end{proof}

\section{\label{sec:LieSignificance}Lie algebraic significance  of generalized SICs}
The connection between SICs and Lie algebra was first studied  by Appleby, Flammia, and Fuchs~\cite{ApplFF11}, who showed that the existence of a SIC in dimension $d$ is equivalent to the existence of a basis for the Lie algebra~$\mathfrak{u}(d)$ such that the structure matrices have a  nice $Q-Q^T$ property. Recently, Appleby, Fuchs, and the author~\cite{ApplFZ14G} generalized the result by proving that the SIC existence is equivalent to the existence of a basis such that each structure matrix is Hermitian and rank $2(d-1)$. The same line of thinking turns out to be useful also for studying generalized SICs.
Here we reveal a cute characterization of  generalized SICs in terms of the structure constants of the Lie algebra.

Given a basis  $L = \{L_j\}$ of Hermitian operators for $\mathfrak{u}(d)$, the  structure constants $C^{L}_{jkl}$ and structure matrices~$C^L_j$  for the basis are defined as
\begin{equation}\label{eq:LieStrucDef}
[L_j,L_k] = \sum_{l} C^{L}_{jkl} L_l,\quad
(C^L_j)_{kl} = C^L_{jkl}.
\end{equation}
Note that the structure constants and structure matrices are pure imaginary.
The structure matrices are Hermitian if and only if the structure constants are completely antisymmetric~\cite{ApplFF11,ApplFZ14G}.

\begin{theorem}
Let $\{\Pi_j\}$ be a minimal IC measurement serving as a basis for the Lie algebra $\mathfrak{u}(d)$. Then the structure constants are completely antisymmetric if and only if $\{\Pi_j\}$ is a generalized SIC.
\end{theorem}
\begin{proof}
If $\{\Pi_j\}$ is a generalized SIC, then  $\{P_j=d\Pi_j\}$ satisfies the equation
\begin{equation}
\tr(P_jP_k)=\alpha\delta_{jk} +\zeta=\alpha\delta_{jk} +\zeta \tr(P_j) \tr(P_k)
\end{equation}
for some positive constants $\alpha$ and $\zeta$. Therefore, the structure constants associated with $\{P_j\}$ are completely antisymmetric according to Lemma 4 in \rcite{ApplFZ14G}, so are the structure constants associated with $\{\Pi_j\}$. If the structure constants are completely antisymmetric, then the same lemma implies that
\begin{equation}
\tr(P_jP_k)=\alpha\delta_{jk} +\zeta \tr(P_j) \tr(P_k)
\end{equation}
for some positive constants $\alpha$ and $\zeta$. Summing over $k$ and applying the identity $\sum_k P_k=d$ yields
\begin{equation}
d\tr(P_j)=\alpha+d^2\zeta \tr(P_j),
\end{equation}
which implies that all $P_j$ have the same trace of 1. So $\tr(P_jP_k)=\alpha\delta_{jk} +\zeta$ and $\{\Pi_j\}$ is a generalized SIC.
\end{proof}

\section{\label{sec:Summary}Summary}
We have identified several  characteristic  merits of  generalized SICs in connection with quantum state tomography and Lie algebra. We showed that among minimal IC measurements  generalized SICs happen to be balanced measurements, measurements whose  tomographic efficiencies are most insensitive to the states under consideration. They are optimal
for   quantum state tomography   given the average purity of measurement outcomes. In addition to establishing the tomographic significance of  generalized SICs, our results are expected to play an important role in studying   tomographic efficiencies of minimal IC measurements in realistic scenarios. In a different vein, when minimal IC measurements are taken as  bases for the Lie algebra of the unitary group, we showed that  generalized SICs are uniquely characterized by the antisymmetry of the associated structure constants.

\bigskip

\section*{Acknowledgements}
It is a pleasure to thank Lin Chen for comments and suggestions.
This work is supported in part by Perimeter Institute for Theoretical Physics. Research at Perimeter Institute is supported by the Government of Canada through Industry Canada and by the Province of Ontario through the Ministry of Research and Innovation.

\bibliographystyle{apsrev}
\bibliography{all_references}

\end{document}